\documentclass[11pt]{article}
\usepackage{amssymb,amsmath,amsthm}
\usepackage{epsfig}
\makeatletter
\def\fps@figure{htbp}
\makeatother
\usepackage{url}
\usepackage{enumitem}

\usepackage[perpage,hang,symbol*]{footmisc} 
\DefineFNsymbols*{lamportnostar}[math]{\dagger\ddagger\S\P\|{\dagger\dagger}{\ddagger\ddagger}}
\setfnsymbol{lamportnostar} 

\makeatletter
\let\@@citation@@=\citation

\renewcommand{\citation}[1]{\@@citation@@{#1}%
\@for\@tempa:=#1\do{\@ifundefined{cit@\@tempa}%
  {\global\@namedef{cit@\@tempa}{}}{}}%
}
\makeatother

\usepackage{hyperref}

\makeatletter
\def\@lbibitem[#1]#2#3\par{%
  \@ifundefined{cit@#2}{}{\@skiphyperreftrue
  \H@item[%
    \ifx\Hy@raisedlink\@empty
      \hyper@anchorstart{cite.#2\@extra@b@citeb}%
        \@BIBLABEL{#1}%
      \hyper@anchorend
    \else
      \Hy@raisedlink{%
        \hyper@anchorstart{cite.#2\@extra@b@citeb}\hyper@anchorend
      }%
      \@BIBLABEL{#1}%
    \fi
    \hfill
  ]%
  \@skiphyperreffalse}%
  \if@filesw
    \begingroup
      \let\protect\noexpand
      \immediate\write\@auxout{%
        \string\bibcite{#2}{#1}%
      }%
    \endgroup
  \fi
  \ignorespaces
  \@ifundefined{cit@#2}{}{#3}}

\def\@bibitem#1#2\par{%
  \@ifundefined{cit@#1}{}{\@skiphyperreftrue\H@item\@skiphyperreffalse
  \Hy@raisedlink{%
    \hyper@anchorstart{cite.#1\@extra@b@citeb}\relax\hyper@anchorend
    }}%
  \if@filesw
    \begingroup
      \let\protect\noexpand
      \immediate\write\@auxout{%
        \string\bibcite{#1}{\the\value{\@listctr}}%
      }%
    \endgroup
  \fi
  \ignorespaces
  \@ifundefined{cit@#1}{}{#2}}
\makeatother

\newtheorem{thm}{Theorem}

\newtheorem{lem}[thm]{Lemma}

\newtheorem{claim}[thm]{Claim}

\newtheorem{obs}[thm]{Observation}

\usepackage{indentfirst}
\usepackage{xspace} 

\def\P{\mbox{\ensuremath{\mathcal P}}\xspace}

\mathcode`l="8000
\begingroup
\makeatletter
\lccode`\~=`\l
\DeclareMathSymbol{\lsb@l}{\mathalpha}{letters}{`l}
\lowercase{\gdef~{\ifnum\the\mathgroup=\m@ne \ell \else \lsb@l \fi}}%
\endgroup

\begin{document}

\title{Complexity of Domination in Triangulated Plane Graphs}
\author{D\"om\"ot\"or P\'alv\"olgyi
\footnote{Institute of Mathematics, E\"otv\"os Lor\'and University (ELTE), Budapest, Hungary
Research supported by the Marie Sk\l odowska-Curie action of the European Commission, under grant IF 660400, and by the Lend\"ulet program of the Hungarian Academy of Sciences (MTA), under grant number LP2017-19/2017.}
}
\maketitle

\begin{abstract}
We prove that for a triangulated plane graph it is {\sc NP}-complete to determine its domination number and its power domination number.
\end{abstract}

\section{Introduction}

Given a graph $G=(V,E)$, for a set of vertices $S$, denote by $\Gamma(S)$ the closed neighborhood of $S$, i.e.,
\[\Gamma(S)=S\cup \{v\in V\mid \exists \; s\in S \textit{ such that } (v,s) \in E\}.\]

$S$ is called a {\em dominating set} if $V=\Gamma(S)$, i.e., every vertex from $V\setminus S$ has a neighbor in $S$.
The size of the smallest dominating set is called the {\em domination number} of $G$ and is denoted by $\gamma(G)$.
A simple graph embedded in the plane without crossing edges is called a {\em triangulated plane graph} if each of its faces (including the other face) is {\em triangular}, i.e., its boundary consists of three edges.
We emphasize that in this paper we only consider {\em simple} graphs, i.e., multiple edges are not allowed.
Garey and Johnson \cite{GJ} have proved that it is {\sc NP}-hard to determine $\gamma(G)$, already for planar graphs.
We extend this result to triangulated planar graphs.

\begin{thm}\label{thm:dom} For a triangulated plane graph $G$ and integer $n$, it is {\sc NP}-complete to determine its domination number, that is, to decide whether $\gamma(G)\le n$.
\end{thm}

Our method also works for the related parameter called {\em power domination number}.
This problem originates from monitoring electrical networks with so-called Phasor Measurement Units; it was first formulated for graphs by Haynes et al. \cite{HHHH}, but we use the (somewhat different) definition given by Brueni and Heath \cite{BH}.
Given a graph $G=(V,E)$, a set of vertices $S$, let $S_1$ be the subset of vertices from $S$ that have exactly one neighbor outside $S$, i.e.,
\[S_1=\{s\in S\mid \exists! \; v\in V\setminus S \textit{ such that } (s,v) \in E\}.\]
The vertices of $S_1$ can {\em propagate} to their neighbors, so we define
\[\Gamma_1(S)=S\cup \Gamma(S_1).\]
The {\em power domination process} starts from any set of vertices $S$, in the first steps applies the $\Gamma$ operator, and then in each following step the $\Gamma_1$ operator, until $\Gamma_1$ stops increasing the size of the set (which happens after finitely many steps in a finite graph).
The set of vertices obtained this way is denoted by
\[\Gamma_P(S)=\Gamma_1(\ldots\Gamma_1(\Gamma(S))\ldots).\]

If $V=\Gamma_P(S)$, then we say that $S$ is a {\em power dominating set} and the size of the smallest such set is the {\em power domination number}, $\gamma_P(G)$, of the graph $G$.
Brueni and Heath \cite{BH} have proved that it is {\sc NP}-hard to determine $\gamma_P(G)$, already for planar graphs.
We extend this result to triangulated planar graphs.

\begin{thm}\label{thm:pdom} For a triangulated plane graph $G$ and integer $n$, it is {\sc NP}-complete to determine its power domination number, that is, to decide whether $\gamma_P(G)\le n$.
\end{thm}

In fact, our construction will be such that either there is an $S$ with $|S|=n$ such that already $V=\Gamma_1(\Gamma(S))$, or $\gamma_P(G)> n$.

For more related literature and background, see the recent works \cite{Aa,DGP}.

\section{Technical claims}

Our reductions are from the {\sc Planar Monotone 3-sat} problem, which was defined and shown to be {\sc NP}-complete in \cite{BK12}.
In this problem the goal is to decide the satisfiability of a conjunctive normal form (CNF), where each clause contains at most $3$ literals, all of which are either negated, or all unnegated, along with a planar embedding of the incidence structure in the following way. (See Figure \ref{fig:pm3sat}.)

\begin{figure}
\begin{center}
\includegraphics[width=.8\textwidth]{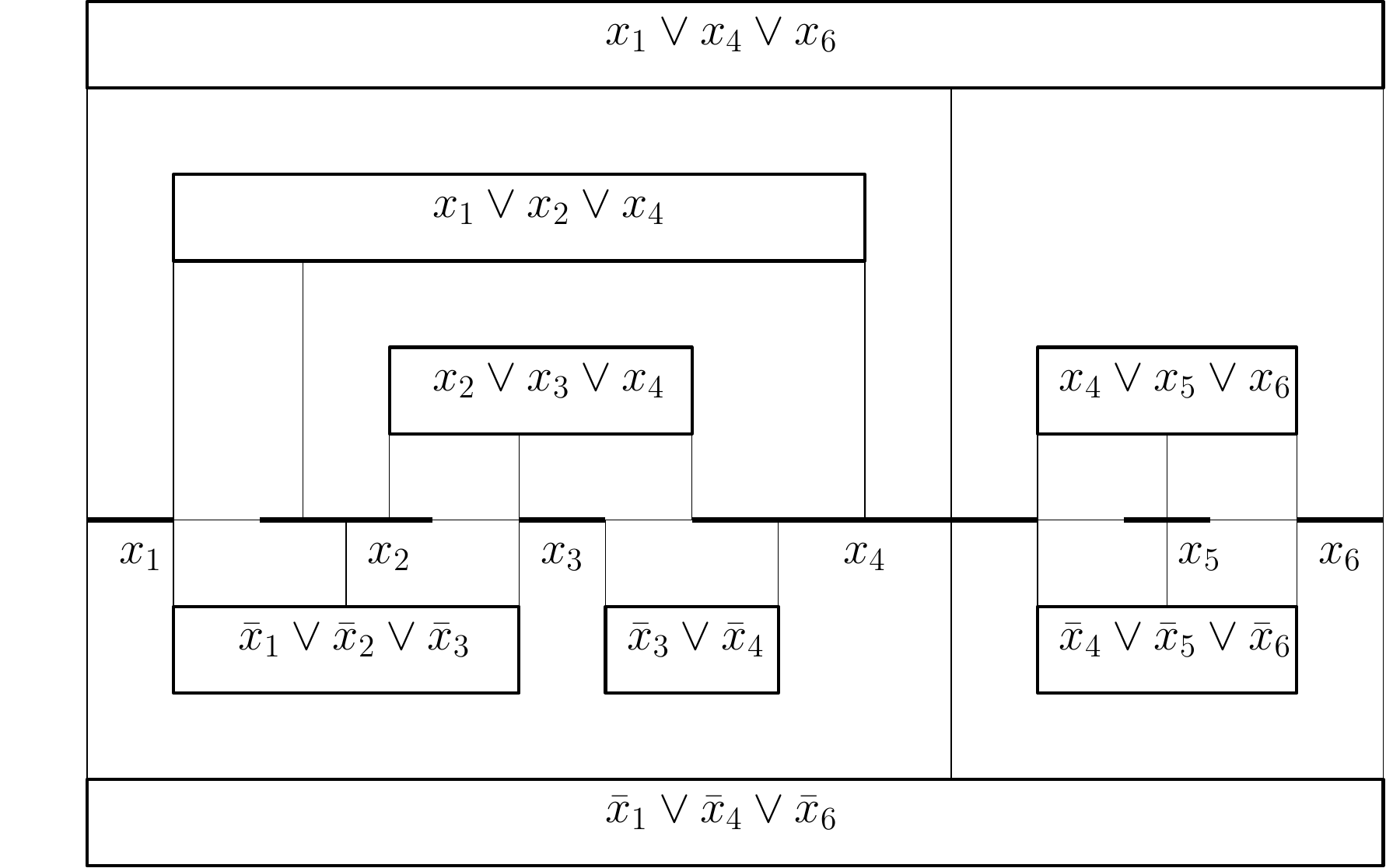}
\caption{Example of a {\sc Planar Monotone 3-sat} input. A satisfying assignment: only $x_4$ is true.}
\label{fig:pm3sat}
\end{center}
\end{figure}

\begin{itemize}
\item Each variable corresponds to an interval in the horizontal line $y=0$; these intervals are pairwise disjoint.
\item Each clause corresponds to an axis-parallel rectangle; these rectangles are pairwise disjoint.
\item If a clause contains only negated (resp.\ unnegated) variables, then its rectangle is entirely contained in the $y<0$ (resp.\ $y>0$) halfplane.
\item Every rectangle is connected to (the intervals corresponding to) the variables contained in (the clause corresponding to) it by a vertical segment, which does not pass through any other rectangles.
\end{itemize}

Note that clauses containing less than $3$ literals are also allowed; 
we are not aware of whether the problem remains {\sc NP}-complete or not if we require that every clause contains exactly $3$ literals (this would slightly simplify our proof).
Note that without requiring monotonicity (and any other special structure) {\sc Planar Exact 3-sat} is {\sc NP}-complete \cite{M83}, even if the planar incidence graph is vertex $3$-connected \cite{K94}.
In our case, however, it seems more likely that the problem always becomes solvable.
This would also follow from a conjecture of Goddard and Henning \cite{GH}, according to which the vertices of any plane triangulation can be $2$-colored such that each vertex is adjacent to a vertex of each color.
(Here we do not go into details about why their conjecture would imply our claim; it involves a triangulation similar to the one that can be found in our main proof.) 

We can, however, suppose that no clause contains exactly $1$ literal, as in this case the formula could be easily simplified.
Moreover, we can also suppose that if a clause contains exactly $2$ literals, then there is no other clause that would contain the same two literals (with the same negations); e.g., $(x_i \vee x_j) \wedge (x_i \vee x_j \vee x_k)$ is equivalent to $(x_i \vee x_j)$.
Because of this, and the properties of the embedding, we can suppose the following.

\begin{obs}\label{obs}
For any two literals there are at most two clauses that contain both of them, and if two such clauses exist, both of them also contains a third literal.
\end{obs}

\bigskip

We will also use the following technical lemma about triangulating plane graphs.

\begin{lem}\label{lem:tri} Suppose that $G=(V,E)$ is a plane graph and $Z\subset V$ is a subset of its vertices such that
\begin{enumerate}[label=(\arabic*)]
\item every vertex $z\in Z$ has at least three neighbors,
\item for a vertex $z\in Z$ and two of the edges adjacent to it, $(z,v)$ and $(z,v')$, that follow each other in the rotation around $z$ in the embedding of $G$, either $(v,v')\notin E$ or $(v,v',z)$ forms a triangular face,
\item if $z,z'\in Z$ are neighbors, then they have exactly two common neighbors, $v,v'\in V$, and $(z,z',v)$ and $(z,z',v')$ are two triangular faces of the embedding,
\item if two vertices $v,v'\in V\setminus Z$ have two common neighbors from $Z$, then they have exactly two common neighbors from $Z$, $z$ and $z'$, and $(v,z,v',z')$ is a face of the embedding of $G$,
\end{enumerate}

\indent
then $G$ can be triangulated by adding only edges that are not adjacent to any vertex in $Z$.
\end{lem}
\begin{proof}
We need to show that if for a vertex $z\in Z$ two of the edges adjacent to it, $(z,v)$ and $(z,v')$, follow each other in the embedding of $G$ in the rotation around $z$, then either $(v,v')\in E$ and $(v,v',z)$ forms a triangular face, or $(v,v')$ can be added as such.
This way the faces around each $z\in Z$ become triangulated and we can triangulate the rest of the graph arbitrarily.

We handle the following cases separately.
\begin{itemize}
\item If $(v,v')\in E$, then because of condition (2) $(v,v',z)$ forms a triangular face.
\item If $v$ or $v'$ is from $Z$, then because of condition (3) $(v,v')\in E$.
\item If $v$ and $v'$ have no other common neighbor from $Z$, then connect them by an edge in the vicinity of the curves of the edges $(v,z)$ and $(z,v')$.
\item If $v$ and $v'$ have another common neighbor from $Z$, then because of condition (4) they have exactly one, $z'\in Z$, and $(v,z,v',z')$ is a face of the embedding of $G$, thus we can divide it by adding the edge $(v,v')$.
\end{itemize}

By repeatedly applying the above, the only condition we could violate is condition (2) by adding the edge $(v,v')$ such that $(v,v',z)$ does not form a triangular face.
But we can add $(v,v')$ to $G$ only in the last two cases, when $v$ and $v'$ have a common neighbor from $Z$, and in each case $(v,v',z)$ forms a triangular face after adding $(v,v')$.
This finishes the proof of Lemma \ref{lem:tri}.
\end{proof}

\section{Proofs of Theorems \ref{thm:dom} and \ref{thm:pdom}}

\begin{proof}[Proof of Theorem \ref{thm:dom}]
The problem is trivially in {\sc NP}, we only have to prove its hardness.

Given an input $\Psi$ to the {\sc Planar Monotone 3-sat} problem on $n$ variables, we transform it into a plane triangulation $G$ such that $\gamma(G)\le n$ if and only if $\Psi$ is satisfiable.
(See Figure \ref{fig:domin1} for the basic graph $G$ obtained from $\Psi$ and Figure \ref{fig:domin3} for the plane triangulation.)

\begin{figure}
\includegraphics[width=\textwidth]{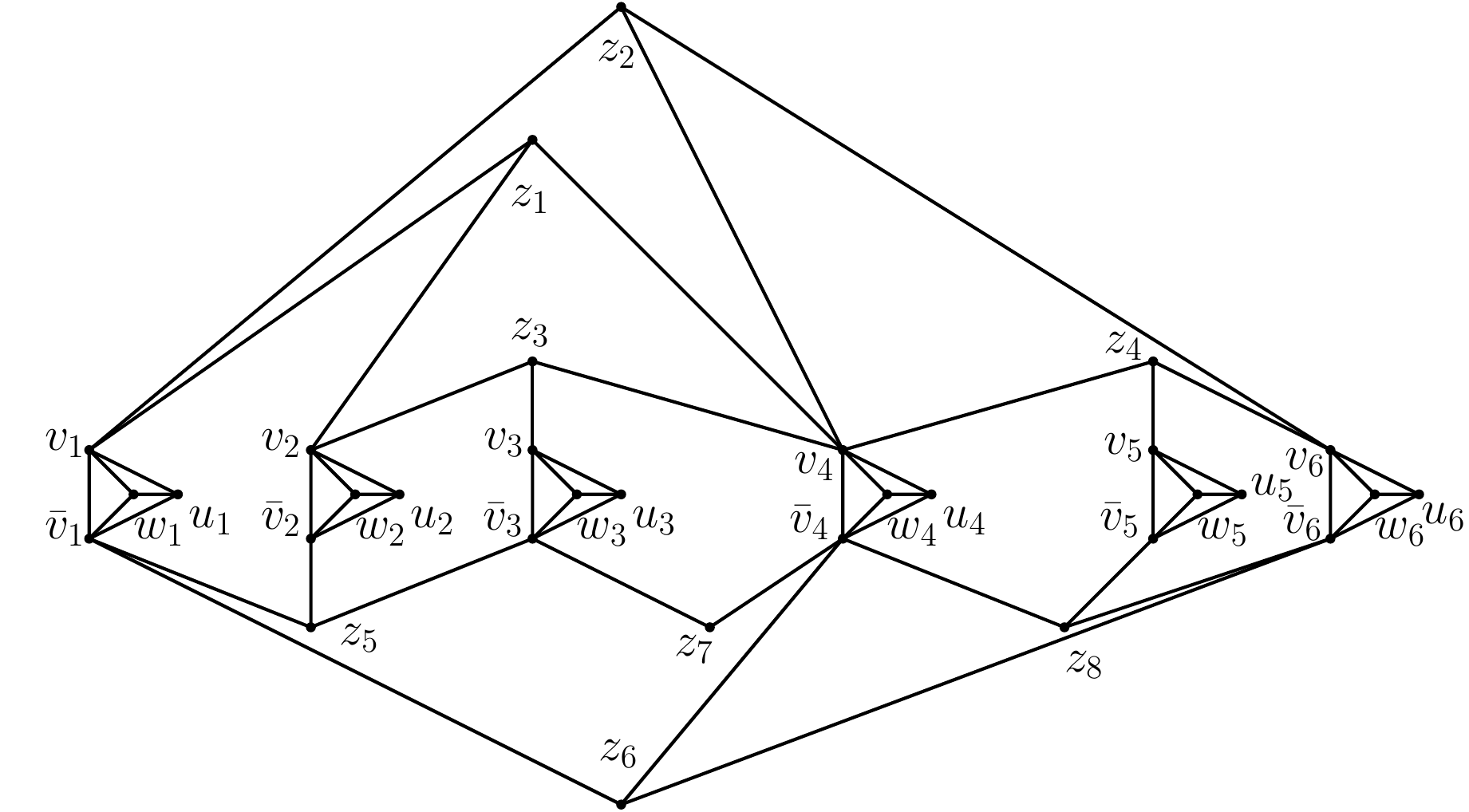}
\caption{Example of graph $G$ used for hardness of domination obtained from {\sc Planar Monotone 3-sat} input. A dominating set of size $6$: $\{v_4,\bar v_1,\bar v_2,\bar v_3,\bar v_5,\bar v_6\}$.}
\label{fig:domin1}
\end{figure}

For each variable $x_i$, $G$ will contain a $K_4$, whose vertices we denote by $v_i,\bar v_i, u_i, w_i$.
The vertex $w_i$ has no other neighbors, which already shows that $\gamma(G)\ge n$, as we must select a vertex from each $K_4$.

For each clause $C_h$ we introduce a vertex, $z_h$, that is connected only to one vertex for each literal it contains;
if $x_i\in C_h$, then we connect $z_h$ to $v_i$, while if $\bar x_i\in C_h$, then we connect $z_h$ to $\bar v_i$.

The graph obtained so-far is obviously planar, now we need the following bound on its domination number. 

\begin{claim} $\gamma(G)=n$ if and only if $\Psi$ is satisfiable.
\end{claim}
\begin{proof} Suppose that $\Psi$ is satisfiable and fix a satisfying assignment.
If $x_i$ is true, we can let $v_i\in S$, and if $x_i$ is false, we can let $\bar v_i\in S$.
This way we have picked a vertex from each $K_4$ corresponding to the variables and since the assignment satisfies $\Psi$, every vertex $z_h$ corresponding to a clause is also dominated.

Suppose that $\gamma(G)=n$ and fix a dominating set $S$ of size $n$.
As $w_i$ needs to be dominated for each $i$, $|S\cap \{v_i,\bar v_i, u_i, w_i\}|=1$.
If $v_i\in S$, we can let $x_i$ be true, if $\bar v_i\in S$, we can let $x_i$ be false, and otherwise we can choose its truth value arbitrarily.
This way each clause is satisfied, as the corresponding vertex $z_h$ had to be dominated.
\end{proof}

\begin{figure}
\includegraphics[width=\textwidth]{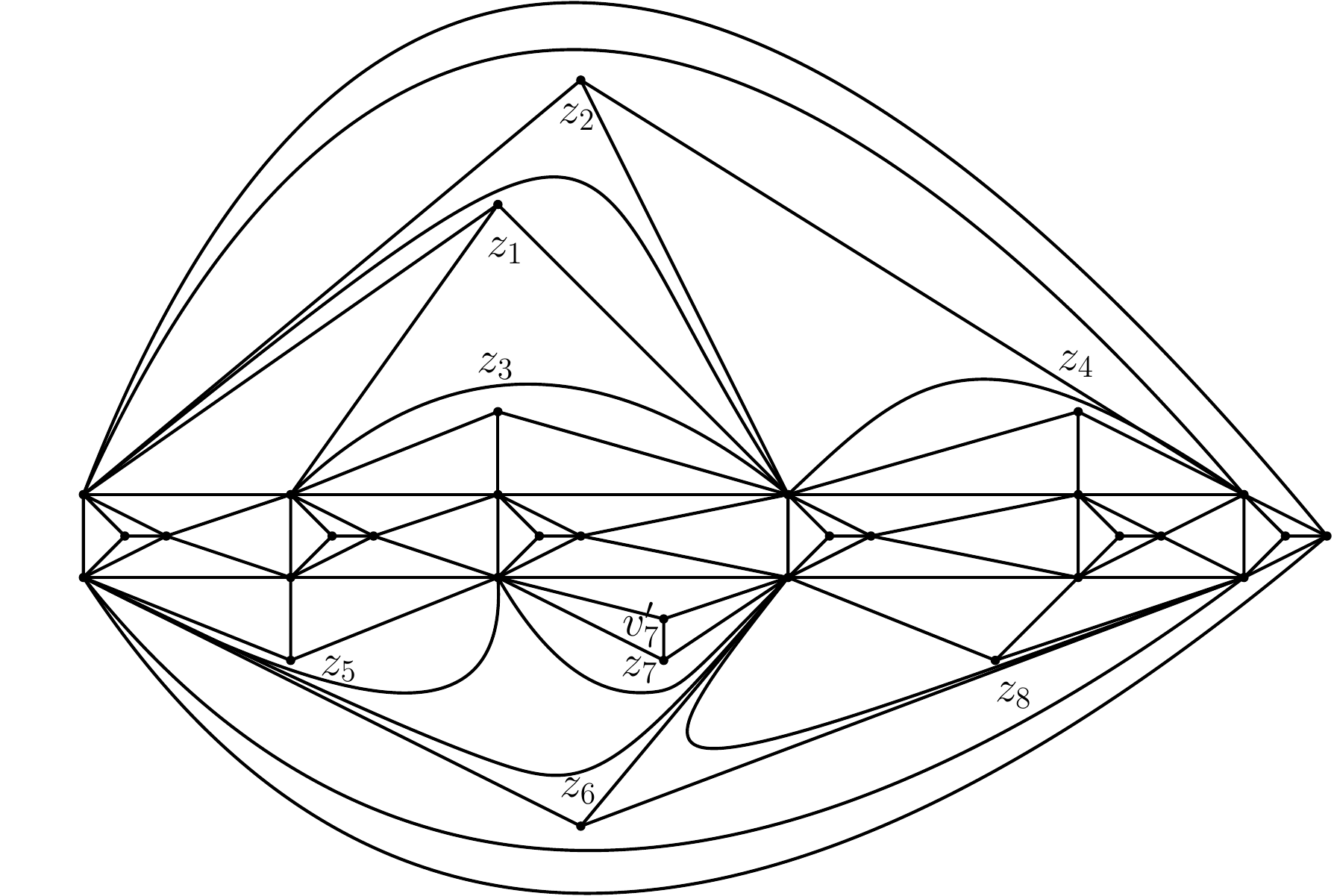}
\caption{Triangulation of $G$ (with vertex $v_7'$ added to the only clause with two variables).}
\label{fig:domin3}
\end{figure}

This already establishes the hardness of the problem for plane graphs; to finish the proof of Theorem \ref{thm:dom}, we only need to show that we can triangulate $G$ without introducing any new neighbors to the $z_h$ vertices.
If each clause of $\Psi$ contains exactly three literals, then this follows by taking the (not necessarily straight-line) ``natural embedding'' of $G$ obtained from the embedding of $\Psi$, and applying Lemma \ref{lem:tri} with $Z$ containing the $z_h$ vertices that correspond to the clauses (it is straight-forward to check that the conditions of Lemma \ref{lem:tri} hold using Observation \ref{obs}).

If $\Psi$ also contains clauses with only two literals, we need to introduce extra vertices to $G$ in the following way.
For each clause with two literals, e.g., $C_h=(x_i \vee x_j)$, we add one extra vertex, $v_h'$, that we connect to $x_i, x_j$ and $z_h$.
Note that this does not change the domination number of $G$, as $v_h'$ is connected to exactly the same vertices as $z_h$, and they are also connected to each other.
But now the conditions of Lemma \ref{lem:tri} hold with $Z$ containing the $z_h$ vertices, thus we can obtain a triangulation, finishing the proof of Theorem \ref{thm:dom}.
\end{proof}

\begin{proof}[Proof of Theorem \ref{thm:pdom}]
\begin{figure}
\includegraphics[width=\textwidth]{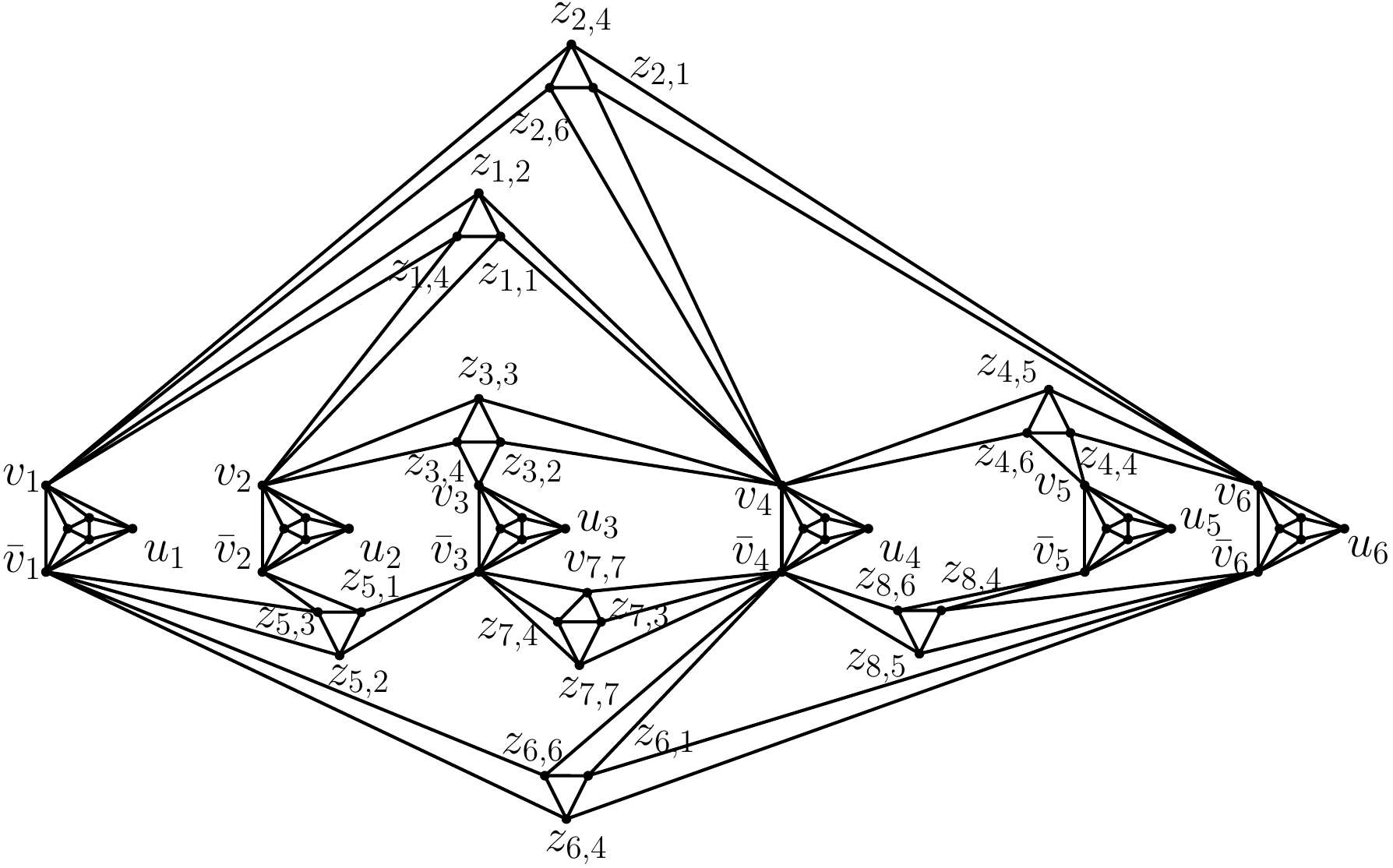}
\caption{Example of graph $G$ used for hardness of power domination obtained from {\sc Planar Monotone 3-sat} input. A power dominating set of size $6$: $\{v_4,\bar v_1,\bar v_2,\bar v_3,\bar v_5,\bar v_6\}$. This graph can be triangulated similarly as on Figure \ref{fig:domin3}.}
\label{fig:pdomin1}
\end{figure}
As in the case of Theorem \ref{thm:dom}, the problem is trivially in {\sc NP}, we only have to prove its hardness.

Given an input $\Psi$ to the {\sc Planar Monotone 3-sat} problem on $n$ variables, we transform it into a plane triangulation $G$ such that $\gamma_P(G)\le n$ if and only if $\Psi$ is satisfiable. (See Figure \ref{fig:pdomin1}.)

For each variable $x_i$, $G$ will contain six vertices, $v_i,\bar v_i, u_i, v_i',\bar v_i', u_i'$, such that they all have edges between them except $(v_i,v_i')$, $(\bar v_i,\bar v_i')$ and $(u_i,u_i')$.
The vertices $v_i',\bar v_i', u_i'$ have no other neighbors among the other vertices of the graph, thus their degrees are $4$.
This already shows that $\gamma(G)\ge n$, as we must select a vertex from each such sextuple\footnote{The six titles won by Barcelona in 2009 (Copa del Rey, La Liga, UEFA Champions League, Supercopa de Espa\~na, UEFA Super Cup and FIFA Club World Cup) have been described as a `sextuple'. This achievement, however, took place over the course of two different Spanish seasons, including a treble in the 2008-09 season. Despite occurring in two seasons, the six titles are still counted as a `sextuple' by many people, because the three added trophies (during the 2009-2010 season) were extra matches of the 2008-2009 treble and all six titles were won in the same calendar year. 
}, otherwise we could not propagate to $v_i',\bar v_i', u_i'$, as each of their neighbors is adjacent to at least two of them.
If, however, we choose any of $v_i,\bar v_i, u_i$ to our initial set $S$, we have $\{v_i,\bar v_i, u_i, v_i',\bar v_i', u_i'\}\subset \Gamma_1(\Gamma(S)) \subset \Gamma_P(S)$.

For each clause with three literals, e.g., $C_h=(x_i \vee x_j \vee x_k)$, we introduce three degree $4$ vertices, $z_{h,i},z_{h,j},z_{h,k}$, that are connected to each other and to two of the literals each; $z_{h,i}$ is connected to $v_j$ and $v_k$, $z_{h,j}$ is connected to $v_i$ and $v_k$, and $z_{h,k}$ is connected to $v_i$ and $v_j$.
(If $C_h$ contained negated literals, than instead of the $v_i, v_j, v_k$ we would use $\bar v_i, \bar v_j, \bar  v_k$.)
Notice that we must select at least one of $v_i, v_j, v_k,z_{h,i},z_{h,j},z_{h,k}$, otherwise we could not propagate to $z_{h,i},z_{h,j},z_{h,k}$, as each of their neighbors is adjacent to at least two of them.
If, however, we choose any of $v_i, v_j, v_k$ to our initial set $S$, we have $\{z_{h,i},z_{h,j},z_{h,k}\}\subset \Gamma_1(\Gamma(S)) \subset \Gamma_P(S)$.


For each clause with two literals, e.g., $C_h=(x_i \vee x_j)$, we introduce four degree $4$ vertices, $z_{h,i},z_{h,j},z_{h,h},v_{h,h}$, that are connected to each other and two additional vertices each: $z_{h,i}$ is connected to $v_j$ and $v_{h,h}$, $z_{h,j}$ is connected to $v_i$ and $v_{h,h}$, and $z_{h,h}$ and $v_{h,h}$ are connected to $v_i$ and $v_j$.
(If $C_h$ contained negated literals, than instead of the $v_i$ and $v_j$ we would use $\bar v_i$ and $\bar v_j$.)
Notice that we must select at least one of $v_i, v_j, z_{h,i},z_{h,j},z_{h,h}, v_{h,h}$, otherwise we could not propagate to $z_{h,i},z_{h,j},z_{h,k}$, as each of their neighbors is adjacent to at least two of them.
If, however, we choose any of $v_i$ or $v_j$ to our initial set $S$, we have $\{z_{h,i},z_{h,j},z_{h,h},v_{h,h}\}\subset \Gamma_1(\Gamma(S)) \subset \Gamma_P(S)$.

The graph obtained so-far is obviously planar, now we need the following bound on its domination number. 

\begin{claim} $\gamma_P(G)=n$ if and only if $\Psi$ is satisfiable.
\end{claim}
\begin{proof} Suppose that $\Psi$ is satisfiable and fix a satisfying assignment.
If $x_i$ is true, we can let $v_i\in S$, and if $x_i$ is false, we can let $\bar v_i\in S$.
This way we have picked a vertex from each sextuple corresponding to the variables and since the assignment satisfies $\Psi$, every vertex corresponding to a clause is power dominated by $S$.

Suppose that $\gamma(G)=n$ and fix a power dominating set $S$ of size $n$.
As we need to pick a vertex from each sextuple for each $i$, $|S\cap \{v_i,\bar v_i, u_i, v_i',\bar v_i', u_i'\}|=1$.
If $v_i\in S$, we can let $x_i$ be true, if $\bar v_i\in S$, we can let $x_i$ be false, and otherwise we can choose its truth value arbitrarily.
This way each clause is satisfied, as the corresponding $z_{h,.}$ vertices had to be power dominated.
\end{proof}

This already establishes the hardness of the problem for plane graphs; to finish the proof of Theorem \ref{thm:pdom}, we only need to show that we can triangulate $G$ without introducing any new neighbors to the $z_h$ vertices.
This follows by taking the ``natural embedding'' of $G$ obtained from the embedding of $\Psi$, and applying Lemma \ref{lem:tri} with $Z$ containing the $z_{h,.}$ vertices that correspond to the clauses (it is straight-forward to check that the conditions of Lemma \ref{lem:tri} hold using Observation \ref{obs}).
\end{proof}

\subsubsection*{Acknowledgment}

I'm thankful to the organizers of the Workshop on Graph and Hypergraph Domination and its participants where this research was started, especially to Paul Dorbec for proposing the problem of the complexity of determining the power domination number in triangulated plane graphs and to Claire Pennarun for useful discussions, providing references and for calling my attention to the related problem about the domination number.


\begin{thebibliography}{99}\fontsize{10}{0}

\bibitem{Aa} A. Aazami, Domination in graphs with bounded propagation: algorithms, formulations and
hardness results, J. Comb. Optim., 19(4): 429--456 (2010).

\bibitem{BK12} M. de Berg and A. Khosravi, Optimal binary space partitions for segments in the plane, International Journal of Computational Geometry \& Applications 22: 187--206 (2012).

\bibitem{BH} D. J. Brueni and L. S. Heath, The PMU placement problem, SIAM Journal on Discrete Mathematics 19(3): 744--761 (2005).

\bibitem{DGP} P. Dorbec, A. Gonz\'alez, C. Pennarun, Power domination in maximal planar graphs, manuscript, arXiv:1706.10047 (2017).

\bibitem{GJ} M. R. Garey and D. S. Johnson, Computers and intractability: a guide to the theory of NP-completeness. W. H. Freeman and Co. San Francisco, California: W. H. Freeman and Co. pp. x+338. ISBN 0-7167-1045-5.

\bibitem{GH} W. Goddard, M. A. Henning, Thoroughly Distributed Colorings, manuscript, arXiv:1609.09684 (2016).

\bibitem{HHHH} T. W. Haynes, S. M. Hedetniemi, S. T. Hedetniemi, and M. A. Henning, Domination in graphs applied to electric power networks, SIAM Journal on Discrete Mathematics 15(4): 519--529 (2002).

\bibitem{M83} A. Mansfield, Determining the thickness of graphs is NP-hard, Proc. Math. Cambridge Phil. Soc., 39 (1983), 9--23.

\bibitem{K94} J. Kratochv\'il, A special planar satisfiability problem and a consequence of its NP-completeness, Discrete Applied Mathematics  52(3), 1994, 233--252.

\end{thebibliography}
\end{document}